\documentclass[12pt,english,draftclsnofoot,onecolumn,peerreview]{IEEEtran}
\usepackage[T1]{fontenc}
\usepackage[latin9]{inputenc}
\usepackage{geometry}
\usepackage{verbatim}
\usepackage{amsthm}
\usepackage{amsmath}
\usepackage{amssymb}
\usepackage{graphicx}
\usepackage{setspace}
\usepackage{esint}
\makeatletter
  \theoremstyle{definition}
  \newtheorem{defn}{\protect\definitionname}
  \theoremstyle{plain}
  \newtheorem{prop}{\protect\propositionname}
\theoremstyle{plain}
\newtheorem{thm}{\protect\theoremname}
  \theoremstyle{plain}
  \newtheorem{lem}{\protect\lemmaname}
 \theoremstyle{definition}
  \newtheorem{example}{\protect\examplename}
  \theoremstyle{plain}
  \newtheorem{cor}{\protect\corollaryname}
  \theoremstyle{remark}
  \newtheorem{rem}{\protect\remarkname}

\usepackage{cite}
\usepackage{wrapfig}
\IEEEoverridecommandlockouts

\@ifundefined{showcaptionsetup}{}{%
 \PassOptionsToPackage{caption=false}{subfig}}
\usepackage{subfig}
\makeatother

\usepackage{babel}
  \providecommand{\definitionname}{Definition}
  \providecommand{\examplename}{Example}
  \providecommand{\lemmaname}{Lemma}
  \providecommand{\propositionname}{Proposition}
  \providecommand{\remarkname}{Remark}
\providecommand{\corollaryname}{Corollary}
\providecommand{\theoremname}{Theorem}

\begin{document}

\title{Downlink Performance Analysis for a Generalized Shotgun Cellular
System}

\author{{\normalsize Prasanna Madhusudhanan$^{*}$, Juan G. Restrepo$^{\dagger}$,
Youjian (Eugene) Liu$^{*}$, Timothy X Brown}$^{*+}${\normalsize ,
and Kenneth Baker}$^{+}${\normalsize \\$^{*}$Department of Electrical,
Computer and Energy Engineering,\\$^{\dagger}$Department of Applied
Mathematics, $^{+}$Interdisciplinary Telecommunications Program\\University
of Colorado, Boulder, CO 80309-0425 USA\\[-0.7em]\{mprasanna, juanga,
eugeneliu, timxb, kenneth.baker\}@colorado.edu}}
\maketitle
\begin{abstract}
In this paper, we analyze the signal-to-interference-plus-noise ratio
(SINR) performance at a mobile station (MS) in a random cellular network.
The cellular network is formed by base-stations (BSs) placed in a
one, two or three dimensional space according to a possibly non-homogeneous
Poisson point process, which is a generalization of the so-called
shotgun cellular system. We develop a sequence of equivalence relations
for the SCSs and use them to derive semi-analytical expressions for
the coverage probability at the MS when the transmissions from each
BS may be affected by random fading with %
\begin{comment}
arbitrary is already any. don't repeat
\end{comment}
{} arbitrary distributions as well as attenuation following arbitrary
path-loss models. For homogeneous Poisson point processes in the interference-limited
case with power-law path-loss model, we show that the SINR distribution
is the same for all fading distributions and is not a function of
the base station density. In addition, the influence of random transmission
powers, power control, multiple channel reuse groups on the downlink
performance are also discussed. The techniques developed for the analysis
of SINR have applications beyond cellular networks and can be used
in similar studies for cognitive radio networks, femtocell networks
and other heterogeneous and multi-tier networks. \end{abstract}

\section{Introduction\label{sec:Introduction}}

The modern cellular communication network is a complex overlay of
heterogeneous networks such as macrocells, microcells, picocells,
and femtocells. The base station (BS) deployment for these network
can be planned, unplanned, or uncoordinated. Even when planned, the
base station (BS) placement in a region typically deviates from the
ideal regular hexagonal grid due to site-acquisition difficulties,
variable traffic load, and terrain. The coexistence of heterogeneous
networks has further added to these deviations. As a result, the BS
distribution appears increasingly irregular as the BS density grows
and is outside standard performance analysis.

Two approaches of modeling have been widely adopted in the literature.
At one end, the BSs are located at the centers of regular hexagonal
cells to form an ideal hexagonal cellular system. At the other end,
the BS deployments are modeled according to a Poisson point process
which we refer to as shotgun cellular system (SCS). In \cite{Brown2000},
we make a connection between these two models on a homogeneous two
dimensional (2-D) plane. It is shown that the signal-to-interference
ratio, $\left(\mathrm{SIR}\right),$ of the SCS lower bounds that
of the ideal hexagonal cellular system and, moreover, the two models
converge in the strong fading regime. Since the BS deployment in the
practical cellular system lies somewhere in between these two extremes,
the analysis of SCSs is important to completely understand the performance
of the cellular networks. Such an analysis holds significance in areas
beyond cellular networks. Wireless LANs, cognitive radios, and ad-hoc
networks are also characterized by irregular deployment of the BSs
\cite{Chandrasekhar2009b,AndJin2008,HunAnd2007,HunAnd2008,JinAnd2007,TruWeb2009,WebAnd2004,WebAnd2005,WebAnd2006a,WebAnd2006c,Weber2004,WebYan2005,WilHam2007,Madhusudhanan2012a,Madhusudhanan2010,Madhusudhanan2011,Madhusudhanan2012}.
It is emphasized that, although the deployment of BSs in practice
is not random, such a study is useful because it allows a thorough
theoretical understanding of many important effects in the strong
fading regime.

A Poisson point process has been adopted in the literature for the
locations of nodes in the study of sensor networks, ad hoc networks
and other uncoordinated and decentralized networks. In the case of
ad-hoc networks, bounds on the transmission capacity have been derived
in several different contexts \cite{Weber2004,TruWeb2009,WebAnd2004,WebAnd2006b,WebAnd2006c,WebYan2004,WebYan2005,AndWeb2007,JinWeb2008}.
Finding the optimal bandwidth partitioning in uncoordinated wireless
networks is considered in \cite{JinAnd2008}. Similar outage probability
analysis in ad-hoc packet radio networks is considered in \cite{Takagi1984,Zorzi1994}.

An underlying assumption in all the previous work is that the density
of transmitters is constant throughout the cellular region, i.e. the
Poisson point process is homogeneous. Such a model does not appropriately
represent practical cellular networks where the BS distribution is
irregular. In this paper, this scenario is incorporated by modeling
the distribution of BSs by a non-homogeneous Poisson point process.
Moreover, the region of interest need not be restricted to $\mathbb{R}^{2}$
as in prior work, and may be $\mathbb{R}^{1}$ or $\mathbb{R}^{3}$.
Furthermore, the performance dependence on the MS location within
the non-homogeneous cellular region is also characterized. Handoff
features and other dynamics are out of the scope of this work.

Finally, most research restricts the interference analysis to popular
fading models like log-normal, Rayleigh, and Rician distributions
and a propagation model that follows the power law path-loss. Here,
the results hold for arbitrary fading distribution and arbitrary path-loss
models. This helps in more accurately modeling the wireless network.
\begin{figure}
\begin{centering}
\includegraphics[bb=25bp 412bp 596bp 668bp,clip,scale=0.75]{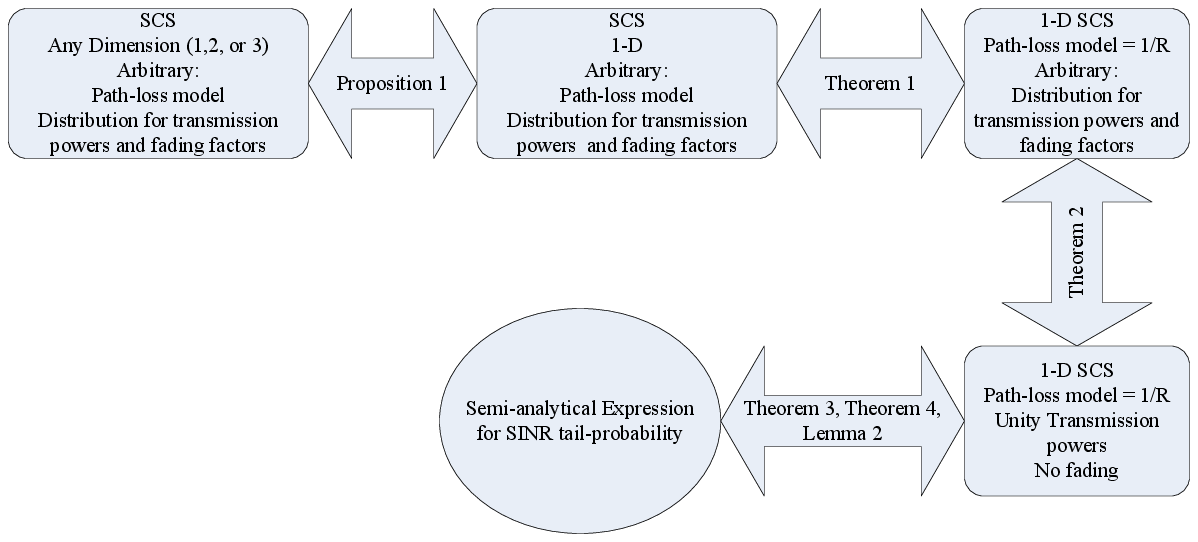}
\par\end{centering}

\caption{\label{fig:contributions1}Contributions of this paper: SINR characterization
for l-D SCSs}
\end{figure}
\begin{figure}
\begin{centering}
\includegraphics[bb=16bp 412bp 542bp 695bp,clip,scale=0.75]{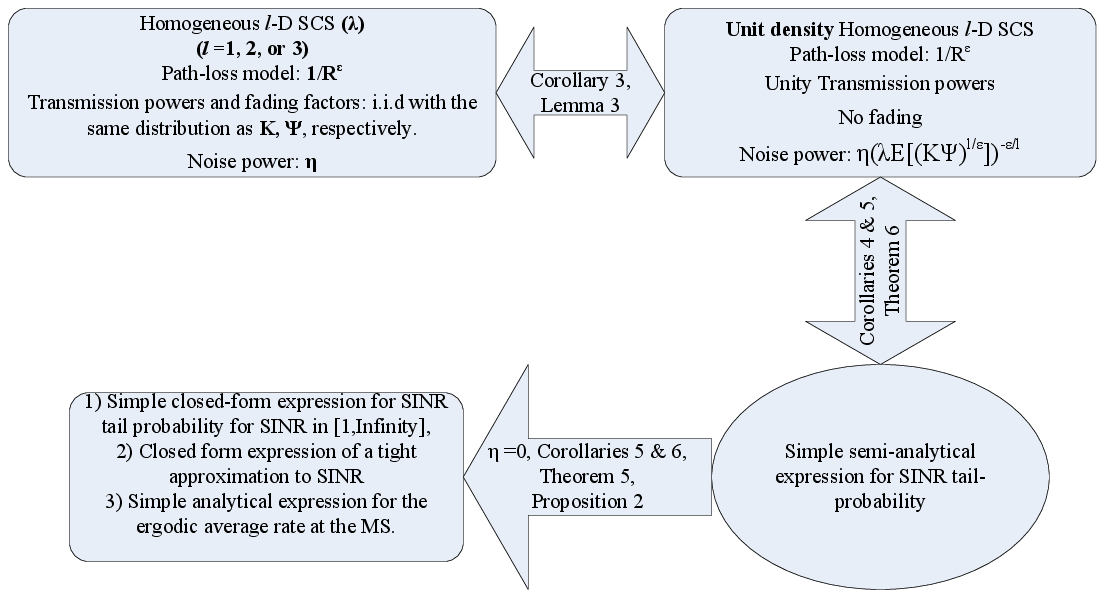}
\par\end{centering}

\caption{\label{fig:contributions2}Contributions of this paper: SINR characterization
for homogeneous $l$-D SCSs}
\end{figure}

The main contributions of this paper are depicted in Figures \ref{fig:contributions1}
and \ref{fig:contributions2}. For general system model where the
BS arrangement is according to a non-homogeneous Poisson point process
in an arbitrary dimension $\left(l=1,2,\ 3\right)$, for arbitrary
path-loss models and arbitrary distributions of the independent and
identically distributed (i.i.d.) random transmission powers and fading
factors, we successively reduce the actual system to a canonical model
that is equivalent in terms of the SINR characteristics, and characterize
the SINR distribution for the simplest equivalent system, thereby
solving the problem for the most general network. For the case of
homogeneous SCSs, which is the most widely used model for random node
locations, we obtain simple closed form characterizations of the SINR
, as well as several insight about the network. This is briefly shown
in Figure \ref{fig:contributions2}.

Applications of the above results in specific wireless communication
scenarios are briefly described in Section \ref{sec:WCommApplications}
followed by the conclusions. Next, the system model and the performance
metric of interest are briefly explained.

\section{System Model\label{sec:modelreview}}

This section describes the various elements used to model the shotgun
cellular system, namely, the BS layout, the radio environment, and
the performance metrics of interest.

\subsection{BS Layout\label{sub:bs_layout}}
\begin{defn}
The \textit{Shotgun Cellular System (SCS) }is a model for the cellular
system in which the BSs are placed in a given $l$-dimensional plane
(typically $l=1,2,\ \mathrm{and}\ 3$) according to a non-homogeneous
Poisson point process on $\mathbb{R}^{l}$ \cite{Kingman1993,Ross1983}.

The intensity function of the Poisson point process is called the
BS density function in the context of the SCS. A 1-D SCS models, for
example, the BS deployments along a highway. A 2-D SCS models planar
BS deployments, and the 3-D SCS models BS deployments in a dense urban
area, or wireless LANs in an apartment building. The 1-D, 2-D and
3-D SCSs are described using the BS density functions $d\left(x\right),$
$d\left(r,\theta\right),$ and $d\left(r,\theta,\phi\right),$ where
$-\infty\le x\le\infty$ represents a point in 1-D, and $r,\ \theta,\ \phi$
are used to represent a point in polar coordinates, in 2-D and 3-D.
\end{defn}
A $l-$D SCS is said to be homogeneous if the BS density function
is a constant over the entire $l$-D space. A homogeneous 2-D SCS
is a common model for the random node placement in many scenarios.

We consider the most general possible description for the wireless
radio environment. Let the received power at a distance $r\ \left(\ge0\right)$
from a given BS be given by
\begin{eqnarray}
P\left(r\right) & = & \left.K\Psi\right/h\left(r\right),\label{eq:receivedPower}
\end{eqnarray}
 where $K$ represents the transmission power and the antenna gain
of the BS, $\Psi$ captures the channel gain, and the function $h\left(\cdot\right)$
represents a path-loss that a signal experiences as it propagates
in the wireless environment. The most commonly used path-loss model
is the power-law path-loss model, $h\left(r\right)=r^{\varepsilon},$
where $\varepsilon$ is called the path-loss exponent.

\subsection{Performance Metric\label{sub:performance_measure}}

In this paper, we focus on the downlink performance of the SCS. In
other words, we are concerned with the signal quality at a mobile-station
(MS) within the SCS. The MS is assumed to be located at the origin
of the $l$-D SCS unless specified otherwise. The signal quality at
the MS is defined as the ratio of the received power from the serving
BS to the sum of the interference powers ($I\ \mathrm{or}\ P_{I}$),
and the background noise power $(\eta)$, and is called the signal-to-interference-plus-noise
ratio (SINR). In an \emph{interference-limited system}, $I\gg\eta$
and the signal quality is the signal-to-interference ratio (SIR).

Using $\left(\ref{eq:receivedPower}\right)$, the SINR at the MS from
a BS at a distance, say $R_{i},$ is
\begin{eqnarray}
\mathrm{SINR} & = & \frac{\left.K_{i}\Psi_{i}\right/h\left(R_{i}\right)}{\underset{j\ne i}{\sum_{j=1}^{\infty}}\left.K_{j}\Psi_{j}\right/h\left(R_{j}\right)+\eta},\label{eq:SINRexpression}
\end{eqnarray}

\noindent where $\{K_{j},\ \Psi_{j}\}_{j=1}^{\infty}$ are independent
and identically distributed (i.i.d) pairs of random variables representing
the transmission power and the channel gain coefficients, respectively,
of the $j^{\mathrm{th}}$ BS, and $\{R_{j}\}_{j=1}^{\infty}$ are
random variables that come from underlying Poisson point process that
governs the BS placement. Also, the MS associates itself with the
BS that has the strongest received signal power (referred to as the
serving BS), and can successfully communicate with this BS, only if
the corresponding SINR exceeds a certain operating threshold, denoted
by $\gamma.$ In this paper, we find the tail probability {[}i.e.
the complementary cumulative density function (c.c.d.f.){]} of the
SINR, which helps characterize an important performance metric for
wireless networks, namely, the coverage probability, i.e. the probability
that a MS is able to successfully communicate with the desired BS.
The following section presents some necessary results that helps simplify
and solve the problem.

\section{SINR Characteristics}

As illustrated in Figure \ref{fig:contributions1}, this section presents
several equivalence relations on BS density, path-loss model, transmission
power and fading that leads to an equivalent canonical SCS model.
Then the equivalence relations are used to simplify the analysis of
the SINR. The equivalence is defined below.
\begin{defn}
\label{def:SCSequivalence}Two SCSs are \textit{equivalent} if the
joint distribution of the powers from all the BSs of a SCS received
at the MS located at the origin is the same as the joint distribution
of the other SCS.
\end{defn}

As a result, if the noise powers are equal, the SINRs at the MSs in
two equivalent SCSs have the same distribution.

The following proposition gives an equivalent 1-D SCS for any $l$-D
SCS. It is a simple consequence of the fact that the path-loss models
considered in this paper is a function of only the distance between
the MS and a BS, not of the orientation.
\begin{prop}
\emph{\label{pro:1DSCSEquivalence}}\textup{An $l$-D SCS, $l=1,\ 2,\ \mathrm{and}\ 3$
is equivalent to a 1-D SCS with a one-sided BS density function $\lambda\left(r\right),\ r\ge0$,
calculated below, if other parameters are the same.}\emph{ }
\begin{itemize}
\item \emph{For a 1-D SCS with density function $d\left(x\right),\ -\infty\le x\le\infty,$
$\lambda\left(r\right)=d\left(r\right)+d\left(-r\right).$}
\item \emph{For a 2-D SCS with density function $d\left(r,\theta\right),$
$\lambda\left(r\right)=\int_{\theta=0}^{2\pi}d\left(r,\theta\right)rd\theta$. }
\item \emph{For a 3-D SCS with density function $d\left(r,\theta,\phi\right),$
$\lambda\left(r\right)=\int_{\theta=0}^{\pi}\int_{\phi=0}^{2\pi}d\left(r,\theta,\phi\right)r^{2}\sin\left(\theta\right)d\theta d\phi.$}
\end{itemize}
\end{prop}

Next, we show the equivalence between SCS's with path-loss model \emph{$\frac{1}{h\left(R\right)}$}
and SCS's with path-loss model\emph{ $\frac{1}{R}$,} using the concepts
of stochastic ordering \cite{Ross1983,ShakedShanthikumar06,ShakedShanthikumar94}.
\begin{thm}
\label{thm:arbitraryPLEquivalenceTheorem}\emph{ If other parameters
are the same, a 1-D SCS with a BS density function $\lambda\left(r\right)$
and path-loss model $\frac{1}{h\left(R\right)}$ is equivalent to
a 1-D SCS with a BS density function $\bar{\lambda}(r)=\lambda(h^{-1}(r))\times\frac{d}{dr}h^{-1}(r),$
and path-loss model $\frac{1}{R},$ where $R$ is the distance between
a BS and the MS, as long as $h\left(r\right),\ r\ge0$ is a monotonically
increasing function with a derivative $h'\left(r\right)>0,\ \forall\ r>0$
and an inverse $h^{-1}\left(r\right).$ As a result, if the noise
powers are the same, the SINRs at the MSs located at the origin in
the two SCSs have the same distribution, i.e. the SINR of (\ref{eq:SINRexpression})
satisfies
\begin{eqnarray}
\left.\mathrm{SINR}\right|_{\lambda(r)} & =_{\mathrm{st}} & \left.\frac{\left.K_{i}\Psi_{i}\right/\tilde{R}_{i}}{\underset{j\ne i}{\sum_{j=1}^{\infty}}\left.K_{j}\Psi_{j}\right/\tilde{R}_{j}+\eta}\right|_{\bar{\lambda}\left(r\right)},\label{eq:arbitraryPLequivalence}
\end{eqnarray}
where $\left\{ \tilde{R}_{j}\right\} _{j=1}^{\infty}$ is the set
of distances of BSs from the MS in the 1-D SCS with BS density function
$\bar{\lambda}\left(r\right)$ and $=_{\mathrm{st}}$ represents the
equivalence in distribution.} \end{thm}
\begin{proof}
See Appendix \ref{sub:proofArbitraryPLEquivalenceTheorem}.
\end{proof}
In the following theorem, we show the equivalence between SCS's with
random transmission power and fading and SCS's with deterministic
transmission power and fading.
\begin{thm}
\emph{\label{thm:arbitraryFadingEquivalenceTheorem}A 1-D SCS with
BS density function $\lambda\left(r\right),$ path-loss model $\frac{1}{R},$
random transmission power $K$ and random fading $\Psi$ that are
i.i.d. across all BSs, is equivalent to another 1-D SCS with a BS
density function $\bar{\lambda}\left(r\right),$ $\frac{1}{R}$ path-loss
model, unity transmission power and unity fading . The above is true
for arbitrary joint distributions of $\left(K,\Psi\right)$ as long
as $\bar{\lambda}\left(r\right)=\mathbb{E}_{K,\Psi}\left[K\Psi\lambda\left(K\Psi r\right)\right]<\infty$
holds for all $r\ge0$,}%
\begin{comment}
\emph{deleted}
\end{comment}
\emph{ where $\mathbb{E}_{K,\Psi}\left[\cdot\right]$ is the expectation
operator w.r.t. $\left(K,\Psi\right).$ The distributions of the SINRs
at the MSs located at the origin of the two SCS's are the same if
the noise powers of the MSs are equal.}\end{thm}
\begin{proof}
See Appendix \ref{sub:proofArbitraryFadingEquivalenceTheorem}.
\end{proof}
Combining Proposition \ref{pro:1DSCSEquivalence}, Theorem \ref{thm:arbitraryPLEquivalenceTheorem}
and Theorem \ref{thm:arbitraryFadingEquivalenceTheorem}, without
loss of generality, we can now restrict our attention to the SINR
characterization of the canonical SCS defined below.
\begin{defn}
A canonical SCS is a 1-D SCS with a BS density function $\lambda\left(r\right),\ r\ge0,$
unity transmission power  and unity fading factors for all BSs in
the SCS, and a path-loss model of $\frac{1}{R}$.
\end{defn}
For a canonical SCS, the BS closest to the origin is the serving
BS and the rest of the BSs contribute to the interference power. The
following is an interesting fact.
\begin{lem}
\emph{\label{cor:SIR-scaling}If the noise powers are the same, the
distributions of SINRs at the MS in canonical SCSs with BS density
function of the form $\frac{1}{a}\lambda(\frac{r}{a})$ are the same
for all $a>0$. In other words, $\left.\mathrm{SINR}\right|_{\lambda(r)}=_{\text{st}}\left.\mathrm{SINR}\right|_{\frac{1}{a}\lambda(\frac{r}{a})}.$}\end{lem}
\begin{proof}
See Appendix \ref{sub:proofSIRScaling}.
\end{proof}
As a result, the appropriate scaling of the BS density function will
not change the p.d.f. of SINR. Next, we derive expressions for the
tail probability of the SINR.
\begin{thm}
\emph{\label{thm:thmcharfn}The tail probability of SINR at the MS
in a canonical SCS, $\mathbb{P}\left(\left\{ \mathrm{SINR}_{\mathrm{canonical}}>\gamma\right\} \right)$
is given by }
\begin{eqnarray}
\mathbb{P}\left(\left\{ \mathrm{SINR}_{\mathrm{canonical}}>\gamma\right\} \right) & = & \begin{cases}
\int_{\omega=-\infty}^{\infty}\Phi_{\frac{1}{\mathrm{\mathrm{SINR}_{\mathrm{canonical}}}}}\left(\omega\right)\left(\frac{1-\exp\left(-\frac{i\omega}{\gamma}\right)}{i\omega}\right)\frac{d\omega}{2\pi}, & \gamma>0\\
1, & \gamma=0
\end{cases},\label{eq:ctoINtailProb}
\end{eqnarray}
\emph{where $\Phi_{\frac{1}{\mathrm{\mathrm{SINR}_{\mathrm{canonical}}}}}\left(\omega\right)$
is the characteristic function of $\frac{1}{\mathrm{\mathrm{SINR}_{\mathrm{canonical}}}}$
given by
\begin{eqnarray}
\Phi_{\frac{1}{\mathrm{\mathrm{SINR}_{\mathrm{canonical}}}}}\left(\omega\right) & = & \mathbb{E}_{R_{1}}\left[\exp\left(i\omega\eta R_{1}\right)\times\Phi_{P_{I}|R_{1}}\left(\left.\omega R_{1}\right|R_{1}\right)\right]\label{eq:charfnitocE2}\\
 & = & \mathbb{E}_{R_{1}}\left[\exp\left(i\omega\eta R_{1}\right)\exp\left(R_{1}\times\int_{u=1}^{\infty}\left(\exp\left(\frac{i\omega}{u}\right)-1\right)\lambda\left(uR_{1}\right)du\right)\right],\label{eq:charfunitoc}
\end{eqnarray}
where $\mathbb{E}_{R_{1}}\left[\cdot\right]$ is the expectation w.r.t.
the random variable $R_{1},$ which is the distance of the BS closest
to the origin, and with the probability density function (p.d.f.)
$f_{R_{1}}\left(r\right)=\lambda\left(r\right)\times\mathrm{e}^{-\int_{s=0}^{r}\lambda\left(s\right)ds},\ r\ge0.$ }\end{thm}
\begin{proof}
See Appendix \ref{sub:proofThmCharFun}.
\end{proof}
Now, we take a minor detour from studying the canonical SCS and consider
a 1-D SCS affected by i.i.d. random  fading factor with unity mean
exponential distribution. For this case, the following theorem gives
a simpler expression for the tail probability of SINR when $\gamma\ge1.$
\begin{thm}
\emph{\label{thm:sinrTailProbGt1}For a 1-D SCS with a BS density
function $\bar{\lambda}\left(r\right),$ $\frac{1}{R}$ path-loss
model, unity transmission power, i.i.d. unity mean exponential random
variable for fading at each BS, the tail probability of SINR for $\gamma\ge1$
is given by
\begin{eqnarray}
\mathbb{P}\left(\left\{ \mathrm{SINR}>\gamma\right\} \right) & = & \int_{r=0}^{\infty}\bar{\lambda}\left(r\right)\exp\left(-\eta\gamma r-\int_{s=0}^{\infty}\frac{\bar{\lambda}\left(s\right)ds}{1+\left(\gamma r\right)^{-1}s}\right)dr.\label{eq:sinrTailProbGt1}
\end{eqnarray}
}\end{thm}
\begin{proof}
See Appendix \ref{sub:proofsinrTailProbGt1}.
\end{proof}

The above result can be used for a canonical SCS under certain conditions.
We briefly investigate this situation for which we define $\mathcal{L}\left(f\left(x\right),s\right)\triangleq\int_{x=0}^{\infty}\mathrm{e}^{-sx}f\left(x\right)dx$
to be the unilateral Laplace transform of the function $f\left(x\right).$
\begin{lem}
\textup{\label{lem:expFadingEquivalence}A canonical SCS with BS
density function $\lambda(r)$ is equivalent to the 1-D SCS considered
in Theorem \ref{thm:sinrTailProbGt1} if there exists a continuous
BS density function $\bar{\lambda}\left(r\right)\ge0$ such that
\begin{eqnarray}
\mathcal{L}\left(\bar{\lambda}\left(x\right),\frac{1}{r}\right) & = & \int_{s=0}^{r}\lambda\left(s\right)ds,\ \forall\ r\ge0.\label{eq:expFadingEquivalence}
\end{eqnarray}
As a result, the tail probability of SINR for such canonical SCS is
equal to $\left(\ref{eq:sinrTailProbGt1}\right).$}\end{lem}
\begin{proof}
The above result is obtained as a consequence of Theorem \ref{thm:arbitraryFadingEquivalenceTheorem}
which says that the two SCSs considered above are equivalent if $\lambda\left(r\right)=\mathbb{E}_{\Psi}\left[\Psi\bar{\lambda}\left(r\Psi\right)\right],\ \forall\ r\ge0,$
where $\Psi$ is the unity mean exponential random variable representing
the fading factors in the latter SCS. By rewriting the expectation
in the above equation as an integral and simplifying, we obtain
\[
\lambda\left(r\right)=\int_{x=0}^{\infty}\frac{d}{dr}\left(\mathrm{e}^{-\frac{x}{r}}\right)\bar{\lambda}\left(x\right)dx\overset{\left(a\right)}{=}\frac{d}{dr}\left(\mathcal{L}\left(\bar{\lambda}\left(x\right),\frac{1}{r}\right)\right),
\]
where $\left(a\right)$ is obtained by exchanging the order of integration
and differentiation, which is valid since $\bar{\lambda}\left(r\right)$
is continuous. Further, the resultant integral can be written in
terms of the Laplace transform of $\bar{\lambda}\left(x\right).$
Using $\left.\mathcal{L}\left(\bar{\lambda}\left(x\right),\frac{1}{r}\right)\right|_{r=0}=0$
as the initial condition, the above differential equation can be solved
to obtain the condition for equivalence between the two SCSs to be
$\left(\ref{eq:expFadingEquivalence}\right).$
\end{proof}

The following shows examples for the existence of BS density functions
$\left(\lambda\left(r\right),\bar{\lambda}\left(r\right)\right)$
that satisfy the condition in $\left(\ref{eq:expFadingEquivalence}\right).$
\begin{example}
\label{exa:expFadingEquivalenceExample1}Polynomial - polynomial equivalence:
The pair $\left(\lambda\left(r\right),\bar{\lambda}\left(r\right)\right)=\left(\alpha_{1}r^{\delta},\alpha_{2}r^{\delta}\right)$
satisfy the condition in $\left(\ref{eq:expFadingEquivalence}\right)$
as long as $\delta+1>0,$ and $\alpha_{1}=\alpha_{2}\Gamma\left(1+\delta\right)>0,$
where $\Gamma\left(\cdot\right)$ is the Gamma function.
\end{example}

\begin{example}
\label{exa:expFadingEquivalenceExample2}Rational - exponential equivalence:
The pair $\left(\lambda\left(r\right),\bar{\lambda}\left(r\right)\right)=\left(\frac{1}{\left(1+\alpha r\right)^{2}},\mathrm{e}^{-\alpha r}\right),\ \forall\ \alpha>0$
satisfy the condition in $\left(\ref{eq:expFadingEquivalence}\right).$
\end{example}

We will see in the following section that the equivalent 1-D BS density
function for the homogeneous $l$-D SCSs are polynomial functions,
and using Example \ref{exa:expFadingEquivalenceExample1} and Theorem
\ref{thm:sinrTailProbGt1}, simple analytical expressions the tail
probability of SINR are obtained.

The results presented in this section can together accurately characterize
the SINR in any arbitrary SCS with arbitrary transmission and channel
characteristics. The semi-analytical expressions presented above might
seem unwieldy at the first glance. But it turns out that several insightful
results can be extracted from this representation for a special class
of SCSs that are practically important and popular in literature.
This special class of SCSs are the homogeneous $l$-D SCSs, $l\in\left\{ 1,2,3\right\} $,
and we dedicate the next section to studying this special class in
detail.

\section{\label{sec:homogeneousLDSCS}Homogeneous $l$-D SCS}

In this section, we focus on the analysis of the homogeneous $l$-D
SCSs with a power-law path-loss model $h\left(R\right)=R^{\varepsilon}$.
The homogeneous $l$-D SCS is the most widely used stochastic geometric
model in the literature for modeling arrangement of node locations.
Especially, its validity in the study of the small-cell networks is
extremely appealing. Moreover, this model has the advantage of being
analytically amenable for a variety of situations that are of great
importance in the modeling and analysis of any type of wireless network.
The results provide several insights about such large-scale networks
that can be applied in the design of actual networks in practice.
Next, we apply the results of the previous section to the case of
the homogeneous $l$-D SCS.
\begin{cor}
\label{cor:mappingHomogeneousSCS}\emph{{[}of Proposition \ref{pro:1DSCSEquivalence}{]}
A homogeneous $l$-D SCS with a constant BS density $\lambda_{0}$
over the entire space is equivalent to the 1-D SCS with a BS density
function $\lambda(r)=\lambda_{0}b_{l}r^{l-1},\ \forall\ r\ge0$, where
$b_{1}=2$, $b_{2}=2\pi$, $b_{3}=4\pi$.}
\end{cor}

This is easily proved by letting $d\left(x\right),$ $d\left(r,\theta\right),$
and $d\left(r,\theta,\phi\right)$ be $\lambda_{0}$ in Proposition
\ref{pro:1DSCSEquivalence}.

For the power-law path-loss model $h\left(R\right)=R^{\varepsilon}$,
we have the following equivalent SCS using Corollary \ref{cor:mappingHomogeneousSCS}
and Theorem \ref{thm:arbitraryFadingEquivalenceTheorem}.
\begin{cor}
\emph{\label{cor:homogeneousSCSPLequivalence}{[}of Theorem \ref{thm:arbitraryFadingEquivalenceTheorem}{]}
A homogeneous $l$-D SCS with BS density $\lambda_{0}$ and path-loss
model $\frac{1}{R^{\varepsilon}}$ is equivalent to the 1-D SCS with
a BS density function $\bar{\lambda}\left(r\right)=\lambda_{0}\frac{b_{l}}{\varepsilon}r^{\frac{l}{\varepsilon}-1},\ r\ge0$
and the path-loss model $\frac{1}{R}.$}
\end{cor}
Next, we characterize the effect of random transmission powers and
fading factors, i.i.d. across BSs in the homogeneous $l$-D SCS.
\begin{cor}
\emph{\label{cor:homogeneousSCSSFEquivalence}{[}of Theorem \ref{thm:arbitraryFadingEquivalenceTheorem}{]}
A homogeneous $l$-D SCS with BS density $\lambda_{0},$ power-law
path-loss model $\left(\frac{1}{R^{\varepsilon}}\right),$ random
transmission powers and fading factors that have arbitrary joint distribution
and are i.i.d. across all the BSs is equivalent to another homogeneous
$l$-D SCS with a BS density $\bar{\lambda}=\lambda_{0}\mathbb{E}\left[\left(K\Psi\right)^{\frac{l}{\varepsilon}}\right],$
same power-law path-loss model $\left(\frac{1}{R^{\varepsilon}}\right),$
unity transmission power and unity fading factor at each BS, where
$K,\ \Psi$ have the same joint distribution as the transmission power
and fading factors of the original homogeneous $l$-D SCS and $\mathbb{E}\left[\cdot\right]$
is the expectation operator w.r.t. $K$ and $\Psi,$ as long as $\mathbb{E}\left[\left(K\Psi\right)^{\frac{l}{\varepsilon}}\right]<\infty.$ }\end{cor}
\begin{proof}
Using Corollary \ref{cor:mappingHomogeneousSCS} and Corollary \ref{cor:homogeneousSCSPLequivalence},
we obtain a 1-D SCS with BS density function $\tilde{\lambda}\left(r\right)=\lambda_{0}\frac{b_{l}}{\varepsilon}r^{\frac{l}{\varepsilon}-1},$
with a path-loss model $\frac{1}{R}.$ Now, from Theorem \ref{thm:arbitraryFadingEquivalenceTheorem},
the equivalent canonical SCS has a BS density function $\hat{\lambda}\left(r\right)=\mathbb{E}\left[\left(K\Psi\right)^{\frac{l}{\varepsilon}}\right]\times\tilde{\lambda}\left(r\right).$
As a result, this can be traced back to the scaling of the BS density
of the original homogeneous $l$-D SCS by $\mathbb{E}\left[\left(K\Psi\right)^{\frac{l}{\varepsilon}}\right].$
\end{proof}

As a result, we can restrict our attention to SINR characterization
when all the BSs of the $l$-D SCS have unity transmission power and
fading factors. Now, we give the expression for the tail probability
of SINR in a homogeneous $l$-D SCS.
\begin{cor}
\label{cor:INtoCcharfn}\emph{{[}of Theorem \ref{thm:thmcharfn}{]}
In a homogeneous $l$-D SCS with a BS density $\lambda_{0}$, unity
transmission power and fading factor at each BS, if the path-loss
exponent of the power-law path-loss model satisfies $\varepsilon>l$,
the characteristic function of the reciprocal of SINR is given by
}\textup{\emph{
\begin{eqnarray}
\Phi_{\frac{1}{\mathrm{SINR}}}\left(\omega\right) & = & \mathrm{E}_{R_{1}}\left[\mathrm{e}^{i\omega\eta R_{1}}\times\mathrm{e}^{\frac{\lambda_{0}b_{l}}{l}R_{1}^{\frac{l}{\varepsilon}}\left(1-_{1}F_{1}\left(-\frac{l}{\varepsilon};1-\frac{l}{\varepsilon};i\omega\right)\right)}\right],\label{eq:homogeneousSCSCharFunSINR}
\end{eqnarray}
}}\textup{where the p.d.f. of $R_{1}$ is }$f_{R_{1}}\left(r\right)=\lambda_{0}\frac{b_{l}}{\varepsilon}r^{\frac{l}{\varepsilon}-1}\cdot\mathrm{e}^{-\lambda_{0}\frac{b_{l}}{l}r^{\frac{l}{\varepsilon}}},\ r\ge0.$\textup{
When $\eta=0,$ the SINR is equivalently the signal-to-interference
ratio (SIR), and
\begin{eqnarray}
\Phi_{\frac{1}{\mathrm{SIR}}}\left(\omega\right) & = & \frac{1}{_{1}F_{1}\left(-\frac{l}{\varepsilon};1-\frac{l}{\varepsilon};i\omega\right)},\label{eq:homogeneousSCSCharFunSIR}
\end{eqnarray}
 where $_{1}F_{1}\left(\dots\right)$ is the confluent hypergeometric
function of the first kind \cite{Mathai2008}. The tail probability
of SINR is given by $\left(\ref{eq:ctoINtailProb}\right).$}\end{cor}
\begin{proof}
From Corollary \ref{cor:homogeneousSCSPLequivalence}, the SINR distribution
is equivalent to the canonical SCS with BS density function $\lambda\left(r\right)=\lambda_{0}\frac{b_{l}}{\varepsilon}r^{\frac{l}{\varepsilon}-1},\ r\ge0.$
Now, by solving for $\left(\ref{eq:charfunitoc}\right),$ in Theorem\emph{
}\ref{thm:thmcharfn}, we obtain $\left(\ref{eq:homogeneousSCSCharFunSINR}\right).$
Further, the expectation in $\left(\ref{eq:homogeneousSCSCharFunSINR}\right)$
reduces to $\left(\ref{eq:homogeneousSCSCharFunSIR}\right)$.
\end{proof}

Due to Corollary \ref{cor:homogeneousSCSSFEquivalence}, the homogeneous
$l$-D SCS satisfies the conditions in Theorem \ref{thm:sinrTailProbGt1}
and hence a simple expression for the tail probability of SINR for
$\gamma\ge1$ can be derived below.
\begin{cor}
\textup{\label{cor:sinrTPGt1HomogeneousSCS}{[}of }\emph{Theorem \ref{thm:sinrTailProbGt1}}\textup{{]}
For a homogeneous $l$-D SCS with BS density $\lambda_{0},$ path-loss
model $\frac{1}{R^{\varepsilon}},\ \varepsilon>l,$ with unity transmission
power and fading factor at each BS, the tail probability of SINR for
$\gamma\ge1$ is
\begin{eqnarray}
\mathbb{P}\left(\left\{ \mathrm{SINR}>\gamma\right\} \right) & = & \int_{r=0}^{\infty}\frac{\lambda_{0}b_{l}r^{l-1}}{\Gamma\left(1+\frac{l}{\varepsilon}\right)}\exp\left(-\eta\gamma r^{\varepsilon}-\frac{\lambda_{0}b_{l}r^{l}\pi\gamma^{\frac{l}{\varepsilon}}}{\varepsilon\Gamma\left(1+\frac{l}{\varepsilon}\right)\sin\left(\frac{l\pi}{\varepsilon}\right)}\right)dr,\label{eq:sinrTPGt1HomogeneousSCS}
\end{eqnarray}
and when $\eta=0,$ the tail probability of SIR is
\begin{eqnarray}
\mathbb{P}\left(\left\{ \mathrm{SIR}>\gamma\right\} \right) & = & \frac{\sin\left(\frac{l\pi}{\varepsilon}\right)\gamma^{-\frac{l}{\varepsilon}}}{\left(\frac{l\pi}{\varepsilon}\right)}=\mathrm{sinc}\left(\frac{l}{\varepsilon}\right)\gamma^{-\frac{l}{\varepsilon}}.\label{eq:sinrTPGt1HomogeneousSCSNoNoise}
\end{eqnarray}
}\end{cor}
\begin{proof}
Due to Corollary \ref{cor:homogeneousSCSSFEquivalence}, the homogeneous
$l$-D SCS is equivalent to another homogeneous $l$-D SCS with the
same path-loss model and transmission powers as the former, and with
a BS density $\frac{\lambda_{0}}{\Gamma\left(1+\frac{l}{\varepsilon}\right)}$
and i.i.d. unity mean exponential random fading factors at each BS.
Using Corollary \ref{cor:homogeneousSCSPLequivalence}, the BS density
function of the 1-D SCS with $\frac{1}{R}$ path-loss model that is
equivalent to the latter homogeneous $l$-D SCS is $\bar{\lambda}\left(r\right)=\frac{\lambda_{0}b_{l}r^{\frac{l}{\varepsilon}-1}}{\varepsilon\Gamma\left(1+\frac{l}{\varepsilon}\right)},\ r\ge0$.
An alternate approach to obtain the expression for $\bar{\lambda}\left(r\right)$
is using Lemma \ref{lem:expFadingEquivalence} and Example \ref{exa:expFadingEquivalenceExample1}.

For the 1-D SCS, Theorem \ref{thm:sinrTailProbGt1} is used to obtain
the expression for the tail probability of SINR to be $\left(\ref{eq:sinrTPGt1HomogeneousSCS}\right)$,
using the identity $\int_{s=0}^{\infty}\frac{s^{\frac{l}{\varepsilon}-1}ds}{1+\left(\gamma r\right)^{-1}s}=\frac{\pi\left(\gamma r\right)^{\frac{l}{\varepsilon}}}{\sin\left(\frac{l\pi}{\varepsilon}\right)}.$
Finally, $\left(\ref{eq:sinrTPGt1HomogeneousSCSNoNoise}\right)$ is
obtained by substituting $\eta=0$ in $\left(\ref{eq:sinrTPGt1HomogeneousSCS}\right)$
and evaluating the outer integral. This completes the proof.
\end{proof}

Using Corollaries \ref{cor:INtoCcharfn} and \ref{cor:sinrTPGt1HomogeneousSCS},
the expression for the tail probability of SINR in a homogeneous $l$-D
SCS with random transmission power and fading factor with an arbitrary
joint distribution that are i.i.d. across the BSs of the SCS can be
obtained by merely scaling the BS density $\lambda_{0}$ with an appropriate
constant that is given in Corollary \ref{cor:homogeneousSCSSFEquivalence}.

The following lemma shows another interesting property of the SINR
distribution in a homogeneous $l$-D SCS.
\begin{lem}
\label{lem:ctoiInvariance}\emph{The SINR distribution in a homogeneous
$l$-D SCS with a constant BS density $\lambda_{0},$ path-loss model
$\frac{1}{R^{\varepsilon}},$ unity transmission power and fading
factor at each BS with a background noise power $\eta$ is the same
as in a homogeneous $l$-D SCS with the same path-loss model, unity
BS density, unity transmission power and fading factor at each BS
and a background noise power }\textup{$\eta\lambda_{0}^{-\frac{\varepsilon}{l}}.$}\emph{
Equivalently,}
\begin{eqnarray}
\left.\mathrm{SINR}\right|_{\left(\lambda_{0},\varepsilon,\eta\right)} & =_{\mathrm{st}} & \left.\mathrm{SINR}\right|_{\left(1,\varepsilon,\eta\lambda_{0}^{-\frac{\varepsilon}{l}}\right)}.\label{eq:ctoINoiseHomogeneous}
\end{eqnarray}
\end{lem}
\begin{proof}
$\left.\mathrm{SINR}\right|_{\left(\lambda_{0},\varepsilon,\eta\right)}\overset{\left(a\right)}{=}\left.\frac{R_{1}^{-\varepsilon}}{\sum_{k=2}^{\infty}R_{k}^{-\varepsilon}+\eta}\right|_{\lambda_{l}\left(r\right)}\overset{\left(b\right)}{=}_{\mathrm{st}}\left.\frac{\left(\alpha R_{1}\right)^{-\varepsilon}}{\sum_{i=2}^{\infty}\left(\alpha R_{i}\right)^{-\varepsilon}+\eta\alpha^{-\varepsilon}}\right|_{\lambda_{l}\left(r\right)}\overset{\left(b\right)}{=}_{\mathrm{st}}\left.\frac{\left(R_{1}^{'}\right)^{-\varepsilon}}{\sum_{k=2}^{\infty}\left(R_{k}^{'}\right)^{-\varepsilon}+\bar{\eta}}\right|_{\frac{1}{\alpha}\lambda_{l}\left(\frac{r}{\alpha}\right)},$
where $\alpha=\lambda_{0}^{\frac{1}{l}}$; $\bar{\eta}=\eta\alpha^{-\varepsilon}$;
$\left(a\right)$ is obtained by expressing SINR in terms of the equivalent
1-D SCS with $\lambda_{l}\left(r\right)=\lambda_{0}b_{l}r^{l-1},\ r\ge0$,
and multiplying numerator and denominator with $\alpha^{-\varepsilon}$;
$\left(b\right)$ follows from Corollary \ref{cor:SIR-scaling}; and
finally, $\left(\ref{eq:ctoINoiseHomogeneous}\right)$ is obtained
by noting that the 1-D SCS with BS density function $\frac{1}{\alpha}\lambda_{l}\left(\frac{r}{\alpha}\right)$
in $\left(b\right)$ corresponds to a homogeneous $l$-D SCS with
BS density 1.
\end{proof}

\begin{figure}
\begin{centering}
\includegraphics[clip,scale=0.65]{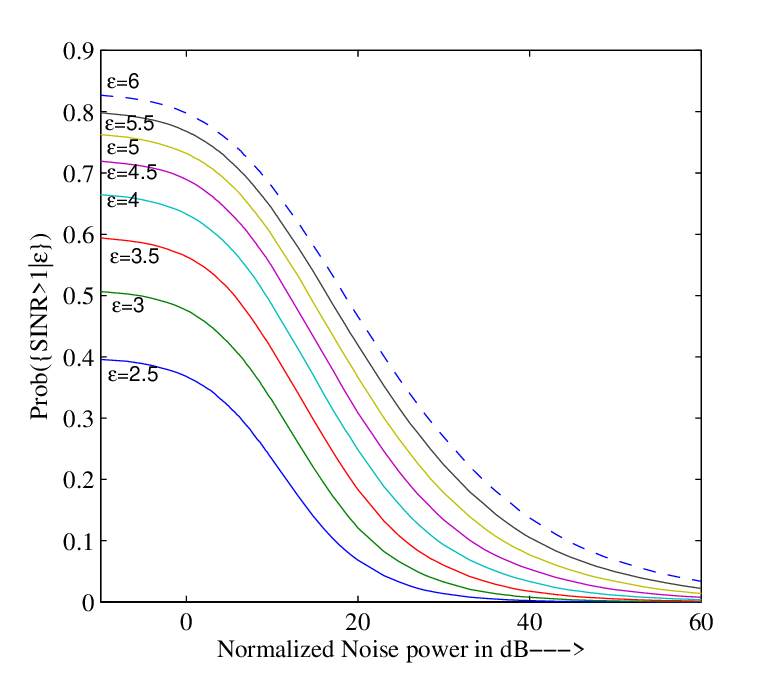}
\par\end{centering}

\caption{\label{fig:cinr_plot}Plot of $\mathrm{Prob}\left(\left\{ \mathrm{SINR}>1\right\} \right)$
vs Noise power}
\end{figure}
Therefore, it is sufficient to analyze a \textit{\emph{homogeneous}}
$l$-D SCS with BS density $\lambda_{0}=1$ and maintain a lookup
table for the tail probability of SINR for different values of the
noise powers and path-loss exponents using $\left(\ref{eq:ctoINtailProb}\right)$.
The lookup table is presented for a \textit{\emph{homogeneous}} 2-D
SCS in Figure \ref{fig:cinr_plot} as a plot of $\mathbb{P}\left(\left\{ \mathrm{SINR}>1\right\} \right)$
against noise powers for different values of path-loss exponents.
Further, in a \textit{\emph{homogeneous}} $l$-D SCS with a high BS
density $\lambda_{0}$, the equivalent noise power $\eta\lambda_{0}^{-\frac{\varepsilon}{l}}$
is small according to Lemma \ref{lem:ctoiInvariance}. Hence, in
an \textit{interference-limited system} (large $\lambda_{0}$), the
signal quality can be measured in terms of SIR. Further remarks on
SIR of a homogeneous $l$-D SCS based on Corollaries \ref{cor:INtoCcharfn}
are given below.

\begin{figure}
\begin{centering}
\includegraphics[scale=0.75]{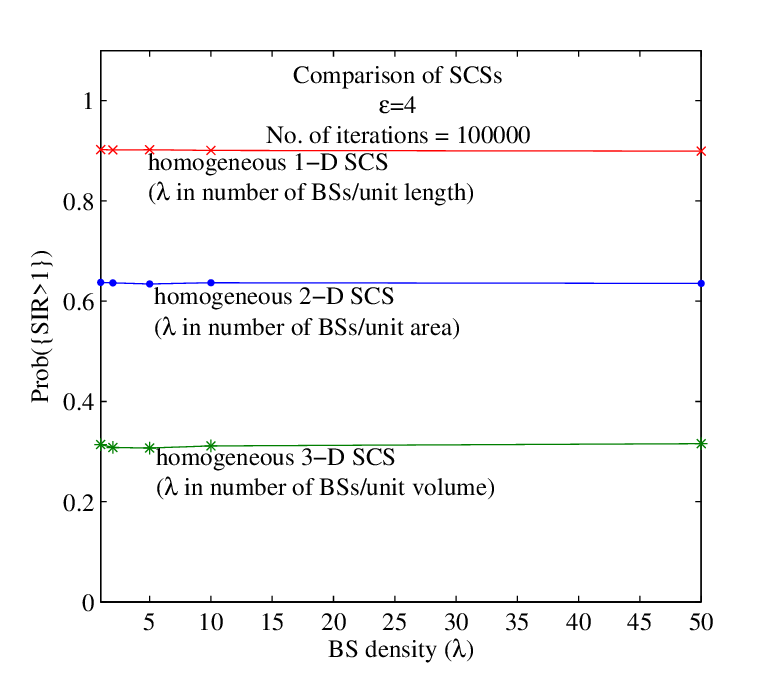}
\par\end{centering}

\caption{\label{fig:uniform_scs_comparison}Simulation showing the invariance
of the performance of homogeneous $l$-D SCS on the BS density. }
\end{figure}

\begin{rem}
\label{rem:ctoiIndependentOfLambda}The characteristic function of
the $\frac{1}{\mathrm{SIR}}$ does not depend on $\lambda_{0}$, and
hence the tail probability of $\mathrm{SIR}$ at a MS in a homogeneous
$l$-D SCS does not depend on $\lambda_{0}$.
\end{rem}

\begin{rem}
\label{rem:sinrInvariance2SFandTransmissionGains}From Corollary \ref{cor:mappingHomogeneousSCS}
and Remark \ref{rem:ctoiIndependentOfLambda}, the tail probability
of SIR is invariant to random transmission powers and fading factors
with arbitrary joint distribution and i.i.d. across the BSs.
\end{rem}

\begin{rem}
\label{rem:ctoiAndEpsilon}The expression for the characteristic function
of $\frac{1}{\mathrm{SIR}}$ for a \textit{homogeneous} 2-D and 3-D
SCS is same as that of a \textit{homogeneous} 1-D SCS with path-loss
exponents $\frac{\varepsilon}{2}$ and $\frac{\varepsilon}{3}$, respectively.
\end{rem}

Remark \ref{rem:ctoiIndependentOfLambda} shows why the tail probabilities
of SIR as a function of BS density for the \textit{homogeneous} 1-D,
2-D and 3-D SCSs in Figure \ref{fig:uniform_scs_comparison} are constant%
\footnote{The simulation results in this section are based on the methods in
Appendix \ref{sec:algo}.%
}. Remark \ref{rem:ctoiAndEpsilon} helps build an intuition of why
the \textit{homogeneous} 1-D SCS has a higher tail probability of
SIR than \textit{homogeneous} 2-D and 3-D SCSs. As the path-loss exponent
decreases, the BSs farther away from the MS have a greater contribution
to the total interference power at the MS, and this leads to a poorer
SIR at the MS and a smaller tail probability. Figure \ref{fig:uniform_scs_comparison}
shows the tail probabilities of SIR in a \textit{homogeneous} 1-D
SCS as a function of the path-loss exponent $\varepsilon$; the squares
$\left(\square\right)$ and the pluses $\left(+\right)$ represent
the values computed analytically and by Monte-Carlo simulations, respectively.
According to Remark \ref{rem:ctoiAndEpsilon}, the same figure can
be used for 2-D and 3-D systems using the scaling of $\frac{\varepsilon}{2},\ \mathrm{and}\ \frac{\varepsilon}{3}$
respectively.

In the following, we present an approximation to SIR based on modeling
the interference due to the strongest few BSs accurately and the interference
due to the rest by their ensemble average. The approximation is expected
to be tight for low BS densities. Due to Remark \ref{rem:ctoiIndependentOfLambda},
the same approximation will be tight for all BS densities. Now, we
define the so-called \textit{few BS approximation} and derive closed
form expressions for the tail probability of SIR at MS in a \textit{homogeneous}
$l$-D SCS for both the SIR regions $\left[0,1\right)$ and $\left[1,\infty\right)$.
\begin{defn}
\textit{The few BS approximation} corresponds to modeling the total
interference power at the MS in a SCS as the sum of the contributions
from the strongest few interfering BSs and an ensemble average of
the contributions of the rest of the interfering BSs.
\end{defn}

Recall that the total interference power is $P_{I}=\sum_{i=2}^{\infty}R_{i}^{-\varepsilon}$,
where $\left\{ R_{i}\right\} _{i=1}^{\infty}$ is the set of distances
of BSs arranged in the ascending order of their separation from the
MS. The arrangement also corresponds to the descending order of their
contribution to $P_{I}$, due to path-loss. In the few BS approximation,
$P_{I}$ is approximated by $\tilde{P_{I}}\left(k\right)=\sum_{i=2}^{k}R_{i}^{-\varepsilon}+\mathbb{E}\left[\left.\sum_{i=k+1}^{\infty}R_{i}^{-\varepsilon}\right|R_{k}\right],$
for some $k$, where $\mathbb{E}\left[\cdot\right]$ is the expectation
operator and corresponds to the ensemble average of the contributions
of BSs beyond $R_{k}$. The SIR at the MS obtained by the few BS approximation
is denoted by $\mathrm{SIR}_{k}.$ The expectation is calculated as
follows.
\begin{lem}
\emph{\label{cor:fewBSMeanCor}For a homogeneous $l$-D SCS, with
BS density $\lambda_{0}$ and $\varepsilon>l$, for $k=1,\ 2,\ 3,\cdots,$}
\begin{eqnarray}
\mathbb{E}\left[\left.\sum_{i=k+1}^{\infty}R_{i}^{-\varepsilon}\right|R_{k}\right] & = & \frac{\lambda_{0}b_{l}R_{k}^{l-\varepsilon}}{\varepsilon-l}.\label{eq:fewBSappMean}
\end{eqnarray}
\end{lem}
\begin{proof}
Firstly, use Corollary \ref{cor:mappingHomogeneousSCS} to reduce
the $l$-D SCS to an equivalent 1-D SCS with BS density function $\lambda(r)=\lambda_{0}b_{l}r^{l-1},\ \forall\ r\ge0$.
Next, given $k$, using the Superposition theorem of Poisson processes,
the original Poisson process is equivalent to the union of two independent
Poisson processes defined in the non-overlapping regions $\left[0,R_{k}\right]$
and $\left(R_{k},\infty\right),$ respectively, with the same BS density
function. Now, using Campbell's theorem \cite[Page 28]{Kingman1993}
to the Poisson process defined in $\left(R_{k},\infty\right),$ we
obtain $\left(\ref{eq:fewBSappMean}\right).$
\end{proof}

The following theorem gives the SIR tail probability approximation,
using $k=2$.
\begin{thm}
\emph{\label{thm:fewBSApproxTheorem}In a homogeneous $l$-D SCS with
BS density $\lambda_{0}$ and path-loss exponent $\varepsilon$, satisfying
$\varepsilon>l$, the tail probability of $\mathrm{SIR}_{2}$ at the
MS is given by
\begin{eqnarray}
\mathbb{P}\left(\left\{ \mathrm{SIR}_{2}>\gamma\right\} \right) & = & \left\{ \begin{array}{ll}
\gamma^{-\frac{l}{\varepsilon}}C_{\frac{\varepsilon}{l}}, & \gamma\geq1\\
1-e^{-u(\gamma)}(1+u(\gamma))+\gamma^{-\frac{l}{\varepsilon}}D_{\frac{\varepsilon}{l}}\left(\gamma\right), & \gamma\le1
\end{array}\right.,\label{eq:few-bs-approx}
\end{eqnarray}
where $C_{\frac{\varepsilon}{l}}=G(0)$ and $D_{\frac{\varepsilon}{l}}\left(\gamma\right)=G(u(\gamma))$
with $G\left(a\right)=\int_{v=a}^{\infty}\frac{ve^{-v}}{\left(1+v\left(\frac{\varepsilon}{l}-1\right)^{-1}\right)^{\frac{l}{\varepsilon}}}dv,$
and $u\left(\gamma\right)\equiv\left(\frac{\varepsilon}{l}-1\right)\left(\frac{1}{\gamma}-1\right)$.}\end{thm}
\begin{proof}
See Appendix \ref{sub:proofFewBSApproxTheorem}.
\end{proof}

The above approximation can be further tightened by recalling that
we already have a simple closed-form expression in $\left(\ref{eq:sinrTPGt1HomogeneousSCSNoNoise}\right)$
for the tail probability of SIR for values in the range $\left[1,\infty\right).$
Hence, the new approximation is as follows
\begin{eqnarray}
\mathbb{P}\left(\left\{ \mathrm{SIR}_{approx}>\gamma\right\} \right) & = & \begin{cases}
\mathbb{P}\left(\left\{ \mathrm{SIR}>\gamma\right\} \right) & ,\ \gamma\ge1\\
\mathbb{P}\left(\left\{ \mathrm{SIR}_{2}>\gamma\right\} \right) & ,\ \gamma\le1
\end{cases},\label{eq:SIRApproximation}
\end{eqnarray}
 where the relevant quantities are obtained from $\left(\ref{eq:sinrTPGt1HomogeneousSCSNoNoise}\right)$
and Theorem \ref{thm:fewBSApproxTheorem}.

\begin{figure}
\subfloat[\label{fig:ctoi_compare-1}]{\begin{raggedright}
\includegraphics[scale=0.65]{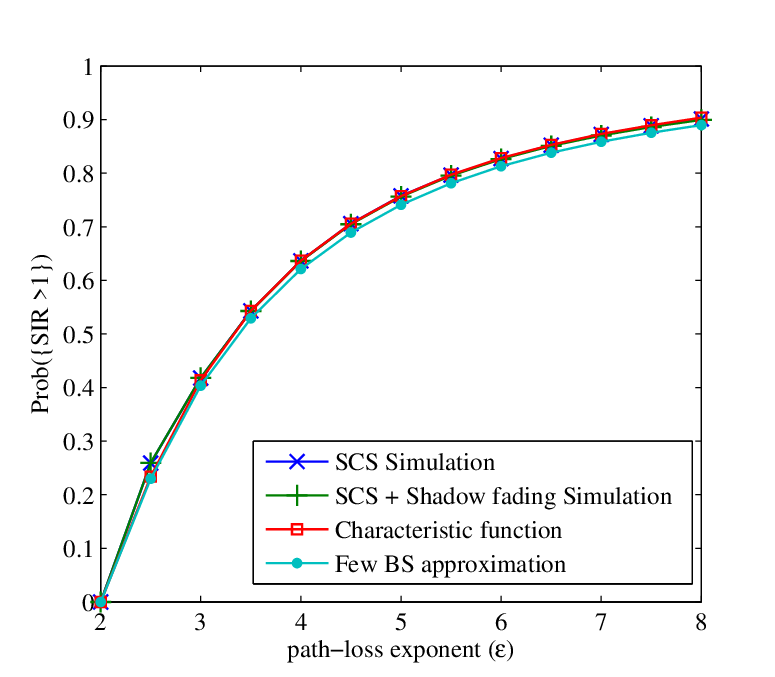}
\par\end{raggedright}

}\subfloat[\label{fig:CompareCtoiCtoi2}]{\raggedright{}\includegraphics[clip,scale=0.65]{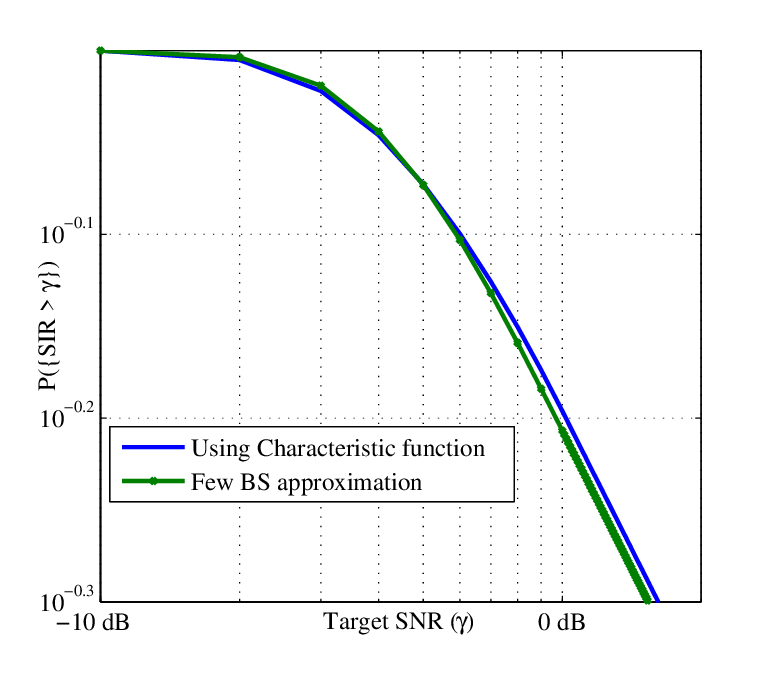}}

\caption{(a) Comparison of Simulations with the analytical results, (b) Homogeneous
2-D SCS: Comparing exact SIR and the few BS approximation for path-loss
$\varepsilon=4$.}
\end{figure}

Notice that $\mathbb{P}\left(\left\{ \mathrm{SIR}>\gamma\right\} \right)=\frac{\mathrm{sinc}\left(\frac{l}{\varepsilon}\right)}{C_{\frac{\varepsilon}{l}}}\mathbb{P}\left(\left\{ \mathrm{SIR}_{2}>\gamma\right\} \right)$
for $\gamma\ge1.$ Figure \ref{fig:ctoi_compare-1} shows that the
few BS approximation $\left(\bullet\right)$ closely follows the actual
behavior $\left(\square\right)$. Figure \ref{fig:CompareCtoiCtoi2}
shows the comparison of the tail probabilities of $\mathrm{SIR}$
(computed using Corollaries \ref{cor:INtoCcharfn} and \ref{cor:sinrTPGt1HomogeneousSCS})
and $\mathrm{SIR}_{2}$ for a \textit{homogeneous} 2-D SCS with path-loss
exponent 4. Notice that the gap between the two tail probability curves
is negligible in the region $\gamma\in\left[0,1\right]$, and further,
both the curves are straight lines parallel to each other in the region
$\gamma\in\left[1,\infty\right)$, when the tail probability is plotted
against $\gamma$, both in the logarithmic scale. This shows that
the few BS approximation characterizes the signal quality in closed
form and is a good approximation for the actual SIR.

Now, having characterized the SIR for the homogeneous $l$-D SCS,
we look closely into what happens when $\varepsilon\le l.$ We will
restrict ourselves to the case when $l=2,$ and the steps are similar
for $l=1$, and $l=3.$
\begin{thm}
\emph{\label{thm:SIRConvergesTo0}A homogeneous 2-D SCS with BS density
$\lambda,$ where the signal decays according to a power-law path-loss
function with a path-loss exponent $\varepsilon\le2,$ the SIR at
the MS is 0 with probability 1.}\end{thm}
\begin{proof}
See Appendix \ref{sub:proofSIRConvergesTo0} for the case $\varepsilon=2.$
From \cite[Corollary 5]{Madhusudhanan2012a}, $\left.\mathbb{P}\left(\left\{ \mathrm{SIR}>\gamma\right\} \right)\right|_{\varepsilon<2}\le\left.\mathbb{P}\left(\left\{ \mathrm{SIR}>\gamma\right\} \right)\right|_{\varepsilon=2}=0,\ \forall\ \gamma\ge0.$
Hence we have proved the above result.
\end{proof}

Note that once we have characterized the SINR distribution, the outage
probability at the MS is known. The event that the MS is in coverage
is given by $\left\{ \mathrm{SINR}>\gamma\right\} ,$ where $\gamma$
is the SINR threshold that the MS should satisfy to be in coverage.
Consequently, the coverage probability, $\mathbb{P}\left(\left\{ \mathrm{SINR}>\gamma\right\} \right)$
is precisely the tail probability of SINR computed at $\gamma.$ Next,
we study the average ergodic reception rate for an MS in coverage.
This quantity, termed as the coverage conditional average rate, is
given by $\mathcal{R}=\mathbb{E}\left[\left.\log\left(1+\mathrm{SINR}\right)\right|\left\{ \mathrm{SINR}>\gamma\right\} \right]$
and is the average of the instantaneous rate achievable at the MS
when the interference is considered as noise. The coverage conditional
average rate at the MS simplifies to the following expression.
\begin{eqnarray*}
\mathcal{R} & = & \log\left(1+\gamma\right)+\int_{t=\Gamma}^{\infty}\frac{\mathbb{P}\left(\left\{ \mathrm{SINR}>t\right\} \right)}{\left(1+t\right)\mathbb{P}\left(\left\{ \mathrm{SINR}>\gamma\right\} \right)}dt.
\end{eqnarray*}

As a result, based on Proposition \ref{pro:1DSCSEquivalence} and
Theorems \ref{thm:arbitraryPLEquivalenceTheorem} - \ref{thm:sinrTailProbGt1},
we can compute the coverage conditional average rate for any SCS.
Specifically, in the interference-limited case, the following proposition
provides the expression for a homogeneous $l$-D SCS and when the
popular power-law path-loss model is assumed. For this case, the SIR
characteristics are invariant to the randomness in the transmission
powers and the fading factors due to Remark \ref{rem:sinrInvariance2SFandTransmissionGains}.
Hence, without loss of generality, we restrict our attention to the
case of constant transmission powers at all BSs and no fading.
\begin{prop}
\emph{\label{pro:avgRate}The ergodic average rate at the MS in a
homogeneous 2-D SCS under the power-law path-loss model, with constant
transmission powers at all BSs and no fading is given by
\begin{eqnarray*}
 &  & \mathcal{R}=\log\left(1+\gamma\right)+\int_{x=\gamma}^{\alpha}\frac{\mathbb{P}\left(\left\{ \mathrm{SIR}>x\right\} \right)}{\mathbb{P}\left(\left\{ \mathrm{SIR}>\gamma\right\} \right)\left(1+x\right)}dx+\alpha^{-\frac{2}{\varepsilon}}\frac{\varepsilon}{2}\cdot_{2}F_{1}\left(1,\frac{2}{\varepsilon};1+\frac{2}{\varepsilon};-\alpha^{-1}\right),
\end{eqnarray*}
where $\alpha=\max\left(\gamma,1\right),$ where $_{2}F_{1}\left(1,\frac{2}{\varepsilon};1+\frac{2}{\varepsilon};-\alpha^{-1}\right)$
is the Gauss hypergeometric function and the probabilities are computed
using $\left(\ref{cor:INtoCcharfn}\right).$ Note that for $\gamma\ge1,$
the middle term drops out.}

\begin{figure}
\subfloat[\label{fig:comparingFadingDistributions}]{\begin{centering}
\includegraphics[clip,scale=0.65]{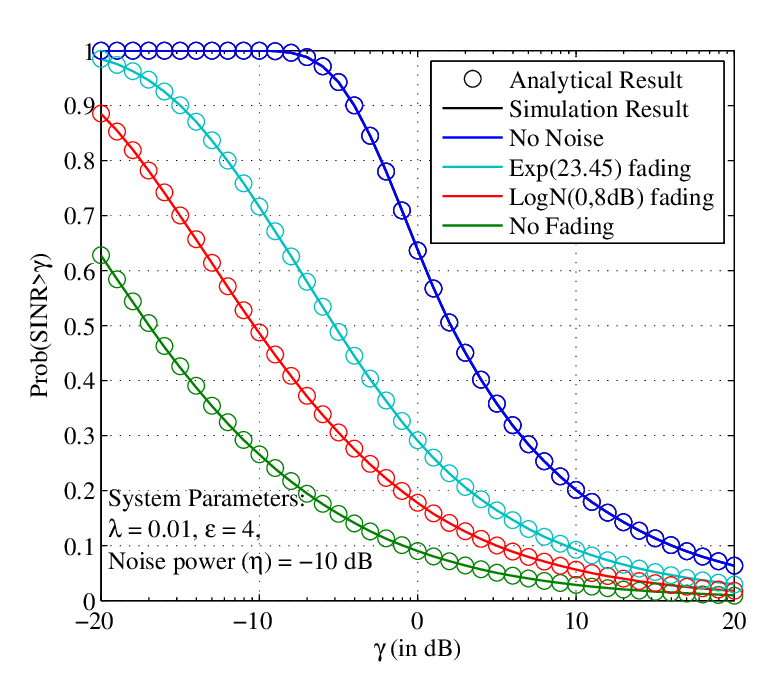}
\par\end{centering}

}\subfloat[\label{fig:comparingExactAndFewBSApprox}]{\begin{centering}
\includegraphics[scale=0.65]{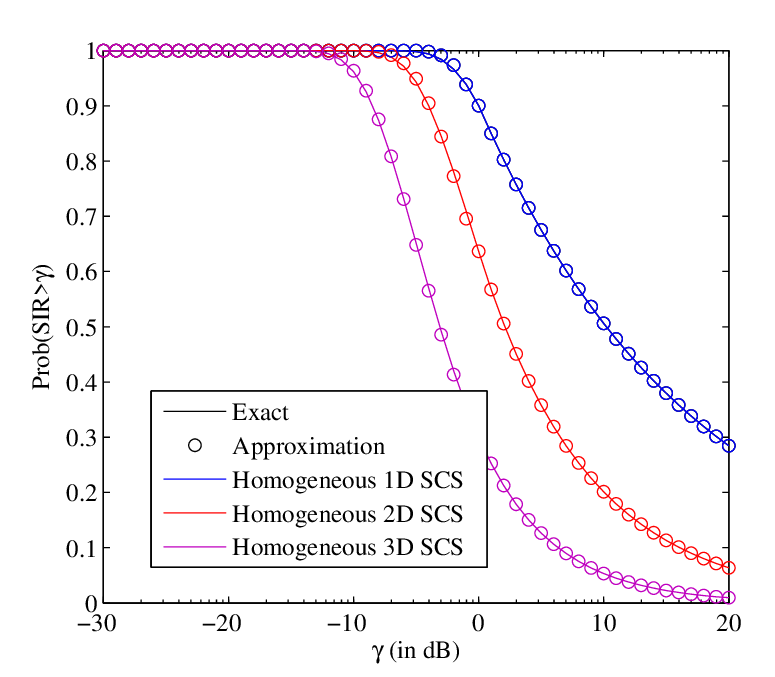}
\par\end{centering}

}

\caption{(a) Comparing the SINR distributions for various fading distributions
and noise profiles, (b) Evaluating the tightness of the few-BS approximation}

\end{figure}

\end{prop}

\section{Numerical Example and Discussion}

In the first example, we consider a homogeneous 2-D SCS with $\lambda=0.01,$
a power-law path-loss model with path-loss exponent 4, and a background
noise power of -10 dB and unity transmission powers. We compare the
SINR tail probabilities for several cases where we vary the distributions
of the fading factors as well as the background noise power. Notice
in Figure \ref{fig:comparingFadingDistributions} that in the case
when there is background noise, the distribution of the fading greatly
affects the SINR performance at the MS. In the presence of the background
noise, the MS sees a better SINR performance when the fading factors
are i.i.d. exponential random variables than when the fading factors
are log-normal random variables, when they have the same mean, and
the SINR performance is far more superior than that without fading.
This is justified by Corollary \ref{cor:homogeneousSCSSFEquivalence}
and Corollary \ref{lem:ctoiInvariance} where the equivalent homogeneous
2-D SCS with unity BS density has an equivalent background noise power
for the log-normal fading case  that is strictly greater than  that
for the exponential fading distributions. Further, in the no noise
case, the SINR performance is invariant to the fading distribution
and is the same as in the no fading case. This is also depicted in
Figure \ref{fig:comparingFadingDistributions}.

In Figure \ref{fig:comparingExactAndFewBSApprox}, we assess the few-BS
approximation for the SIR characterization in the homogeneous $l$-D
SCS. This figure shows that the SIR approximation derived in Section
\ref{sec:homogeneousLDSCS} based on the few-BS approximation (Equation
(\ref{eq:SIRApproximation})) closely follows the exact SIR characterization.
Moreover, this relationship holds for a wide range of scenarios of
interest such as for arbitrary fading and transmission power distributions,
and for all BS densities. In the following section, we discuss the
usage of the results obtained thus far in the analysis of other useful
wireless communication scenarios.

\section{Applications in wireless communications\label{sec:WCommApplications}}

We discuss several scenarios where the wireless communication systems
are modeled by the \textit{homogeneous} $l$-D SCS with BS density
$\lambda_{0}$, where $l$ = 1, 2, and 3 correspond to highway, suburban,
and dense urban deployments, respectively.

\subsubsection*{BSs with sectorized antennas}

In this example, we give a practical scenario where the transmission
powers of the BSs are i.i.d. random variables. For example, consider
the case where each BS has an ideal sectorized antenna with gain $G$
and beam-width $\theta$, such that BS's antenna faces the MS with
probability $\frac{\theta}{2\pi}$, in which case $K_{i}=G$, and
otherwise $K_{i}=0$. In this case, in the absence of fading, from
Corollary \ref{cor:homogeneousSCSSFEquivalence}, $\overline{\lambda}_{0}=\lambda_{0}G^{\frac{2}{\varepsilon}}\frac{\theta}{2\pi}$
is the BS density of the equivalent \textit{homogeneous} $l$-D SCS.

\subsubsection*{Multiple Access Techniques}

Next, we study the signal quality at the MS in a cellular system employing
different multiple access techniques. For example, in a code division
multiple access (CDMA) system, the goal is to maintain a constant
voice signal quality at the MS, which is done by power control. This
goal is achievable by having the serving BS increase its transmission
power by $\alpha=\gamma\mathrm{SIR}^{-1}$ , where $\alpha$ is the
power control factor or the processing gain, SIR is the instantaneous
signal quality at the MS, and $\gamma$ is the desired constant signal
quality. In this formulation, $\alpha$ for each BS is a random variable
and in general, the $\alpha$'s of nearby BSs are correlated. But
if the correlation is small, the SIR distribution computed here enables
radio designers to approximately model the power needs to communicate
with a MS in a SCS. In another formulation, if $\alpha$ is a constant
factor by which the power of the serving BS is improved, its effect
on the tail probability SIR at the MS is obtained by straightforward
manipulations as $\mathbb{P}\left(\left\{ \left.\alpha\times\mathrm{SIR}>\gamma\right|\varepsilon,l\right\} \right)=\mathrm{sinc}\left(\frac{l}{\varepsilon}\right)\left(\frac{\gamma}{\alpha}\right)^{-\frac{l}{\varepsilon}}\ \mbox{if }\gamma>\alpha.$

Then, consider frequency division multiple access (FDMA) and time
division multiple access (TDMA) based cellular systems. Let the available
spectrum (in frequency for FDMA and in time-slots for TDMA) be divided
into $N$ channel reuse groups (CG), and indexed as $k=1,2,\mathrm{\ \cdots,}\ N$
. Then, each BS is assigned one of the $N$ CGs, such that the $k^{\mathrm{th}}$
CG is assigned with probability $p_{k}$. In such a system, the MS
chooses a CG that corresponds to the best SIR; the BS in the CG that
corresponds to the strongest received power is the \textit{desired}
BS, and the MS chooses it as the serving BS. The SIR at the MS in
such a SCS is of interest to us. Note that this \textit{homogeneous}
$l$-D SCS is equivalent to $N$ independent \textit{homogeneous}
$l$-D SCSs with constant BS densities $\lambda_{0}p_{1},\cdots,\lambda_{0}p_{N}$,
by the properties of Poisson point processes. The tail probability
of\textbf{ }SIR at the MS in such a system is given by $\mathbb{P}\left(\left\{ \left.\mathrm{SIR}>\gamma\right|\varepsilon,N\right\} \right)=1-\left[1-\mathbb{P}\left(\left\{ \left.\mathrm{SIR}>\gamma\right|\varepsilon\right\} \right)\right]^{N}$,
where the tail probability on the right hand side is computed using
Corollary \ref{cor:INtoCcharfn}.

\subsubsection*{Cognitive Radios}

In cognitive radio technology, the cognitive radio devices (or \textit{secondary
users}) opportunistically operate in licensed frequency bands occupied
by \textit{primary users}. The interference caused by \textit{secondary
user} transmissions is harmful for \textit{primary users} operation,
and is not acceptable beyond certain limits. Studying the nature of
these interferences and formulating methods for addressing them has
been an active area of research. The results in this paper are a rich
source of mathematical tools for such studies. In \cite{Madhusudhanan2010},
we have extensively applied the results developed here to understand
the role of cooperation between the \textit{secondary users} in ensuring
that the interference caused by the \textit{secondary users} are within
the acceptable limits. The \textit{secondary users} are modeled analogous
to BS placement in \textit{homogeneous} 1-D and 2-D SCS, and the tail
probability of $\frac{C}{I}$ at the \textit{primary user} is characterized.
Further, in the context of radio environment map (REM, \cite[and references therein]{Madhusudhanan2010}),
we have highlighted the practical significance of the study of 1-D
SCS.

\subsubsection*{Overlay Networks}

The modern cellular communication network is a complex overlay of
heterogeneous networks, such as macrocells, microcells, picocells
and femtocells. This complex overlay network is seldom studied as
is, due to the analytical intractabilities. In \cite{Chandrasekhar2009b,Xia2010a},
cellular systems consisting of macrocell and femtocell networks are
analyzed. Using the results in our paper, the cumulative effect of
all the networks constituting the overlay network, on the signal quality
at the MS can be studied. A detailed study on this is set aside as
a future work, while the preliminary results are presented in \cite{Madhusudhanan2011,Madhusudhanan2012}.
Other efforts on the downlink performance characterization for heterogeneous
networks can be found in \cite{Dhillon2012,Dhillon2011b,Dhillon2011a,Dhillon2012,Jo2011,Mukherjee2012,Mukherjee2012a}.

\section{Conclusions\label{sec:conclusions}}

In this paper, we study the characterization of the SIR and SINR at
the MS in shotgun cellular systems where a SCS is defined as a cellular
system where the BS deployment in a given region is according to a
Poisson point process. A sequence of equivalent SCSs are derived to
show that it is sufficient to study the canonical SCS that has unity
transmission power and unity fading factors, and a path-loss model
of $\frac{1}{R}$. Analytical expressions for the tail probabilities
of the SIR and SINR at the MS are obtained for 1-D, 2-D and 3-D SCSs,
where the 1-D, 2-D and 3-D SCS are mathematical models for BS deployments
along the highway (1-D), in planar regions (2-D) and in urban areas
(3-D), respectively. Further, a closed form expression for the tail
probability of SIR is derived for the homogeneous cases of 1-D, 2-D
and 3-D SCS.  The results are applicable for general fading distributions
and arbitrary path-loss models. This makes the results useful for
analyzing many different wireless scenarios that are characterized
by uncoordinated deployments. The application of the results has been
demonstrated in the study of the impact of cooperation between cognitive
radios in the low power primary user detection and can be found in
\cite{Madhusudhanan2010}, and in the study of heterogeneous networks
in \cite{Madhusudhanan2011}. Future work will further explore the
applications of the SCS model in the context of indoor femtocells,
cognitive radios, and multi-tier or overlay networks.

\appendix

\subsection{\label{sub:proofArbitraryPLEquivalenceTheorem}Proof for Path-loss
Equivalence Theorem (Theorem \ref{thm:arbitraryPLEquivalenceTheorem})}

Let $\bar{R}=h(R)$ be the equivalent BS location. Using the \textit{Mapping
Theorem} in \cite{Kingman1993}, BS with locations $\bar{R}$ is also
a Poisson point process, whose density is obtained below. For any
non-homogeneous 1-D Poisson point process, $\mathbb{E}\left[N\left(r+s\right)-N\left(r\right)\right]=\int_{r}^{r+s}\lambda(z)dz$
is the expected number of occurrences in the interval $\left(r,r+s\right)$.
Thus,
\begin{eqnarray}
\mathbb{E}\left[N\left(r+s\right)-N\left(r\right)\right] & = & \mathbb{E}\left[\mbox{Number of BSs with }\bar{R}\in\left(r,r+s\right)\right]\label{eq:app1Step1}\\
 & = & \mathbb{E}\left[\mbox{Number\ of\ BSs\ with\ }R\in\left(h^{-1}\left(r\right),h^{-1}\left(r+s\right)\right)\right]\nonumber \\
 & \mathrm{=} & \int_{z=h^{-1}\left(r\right)}^{h^{-1}\left(r+s\right)}\lambda\left(z\right)dz\ =\ \int_{z=r}^{r+s}\frac{\lambda\left(h^{-1}\left(z\right)\right)}{h'\left(h^{-1}\left(z\right)\right)}dz.\nonumber
\end{eqnarray}
 Hence, the 1-D SCS with path-loss model $\frac{1}{h(R)}$ and a BS
density function $\lambda(r)$ is equivalent to the 1-D SCS with path-loss
model $\frac{1}{R}$ and BS density function $\bar{\lambda}\left(r\right)$.

\subsection{\label{sub:proofArbitraryFadingEquivalenceTheorem}Proof for Arbitrary
Fading Equivalence Theorem (Theorem \ref{thm:arbitraryFadingEquivalenceTheorem})}

Let $\bar{R}=R\left(K\Psi\right)^{-1}$, where $R$ is the random
variable representing the distance from the MS to a BS in the 1-D
SCS with a BS density function $\lambda\left(r\right)$, $K,\ \Psi$
are the transmission power and the fading factor corresponding to
the BS, respectively, and $\bar{R}$ is the corresponding equivalent
distance. Using the \textit{product space representation} and \textit{Marking
Theorem} in \cite{Kingman1993}, $\bar{R}$ also corresponds to the
1-D SCS with a BS density function derived following $\left(\ref{eq:app1Step1}\right):$
\[
\mathbb{E}\left[N\left(r+s\right)-N\left(r\right)\right]\overset{\left(a\right)}{=}\mathbb{E}_{K,\Psi}\left[\int_{rK\Psi}^{\left(r+s\right)K\Psi}\lambda\left(z\right)dz\right]\overset{\left(b\right)}{=}\int_{r}^{\left(r+s\right)}\mathbb{E}_{K,\Psi}\left[K\Psi\lambda\left(K\Psi z\right)\right]dz,
\]
 where $\mathrm{\left(a\right)}$ is obtained by rewriting the expectation
with respect to each realization of $\Psi$ and $K$, and $\mathrm{\left(b\right)}$
is obtained by exchanging the order of integration and expectation
in $\left(b\right)$ as $\mathbb{E}_{K,\Psi}\left[K\Psi\lambda\left(K\Psi z\right)\right]<\infty.$
Hence, $\bar{R}'s$ corresponds to the 1-D SCS with a BS density function
$\bar{\lambda}\left(r\right)=\mathbb{E}_{K,\Psi}\left[K\Psi\lambda\left(K\Psi r\right)\right].$

\subsection{\label{sub:proofSIRScaling}Proof for Corollary \ref{cor:SIR-scaling}}

Let $\left\{ R_{k}\right\} _{k=1}^{\infty}$ correspond to the 1-D
SCS with BS density function $\lambda(r).$ Then, since the ordered
base station locations $R_{k}$'s are determined by inter-base station
distances, it follows that $\begin{array}{c}
\left.\mathrm{SINR}\right|_{\lambda(r)}\end{array}\overset{\left(a\right)}{=}\left.\frac{(aR_{1})^{-1}}{\sum_{k=2}^{\infty}(aR_{k})^{-1}+\eta}\right|_{\lambda(r)}\overset{\left(b\right)}{=}_{\text{st}}\left.\frac{\left(R_{1}^{'}\right)^{-1}}{\sum_{k=2}^{\infty}\left(R_{k}^{'}\right)^{-1}+\eta}\right|_{\frac{1}{a}\lambda(\frac{r}{a})},$ where the SINR expression is obtained using $\left(\ref{eq:SINRexpression}\right)$
with $h\left(R\right)=R$, $\left(a\right)$ is obtained by multiplying
the numerator and denominator by $\frac{1}{a},\ a>0$ $\left(b\right)$
follows from from the properties of Poisson point processes. Further,
$\left\{ R_{k}^{'}\right\} _{k=1}^{\infty}$ in $\left(b\right)$
correspond to 1-D SCS with BS density $\frac{1}{a}\lambda\left(\frac{r}{a}\right),$
$a>0$.

\subsection{\label{sub:proofThmCharFun}Proof for the Tail Probability of SINR
(Theorem \ref{thm:thmcharfn})}

The following are the sequence of step to derive the expression in
$\left(\ref{eq:ctoINtailProb}\right).$
\begin{eqnarray*}
\mathbb{P}\left(\left\{ \mathrm{SINR}_{\mathrm{canonical}}>\gamma\right\} \right) & = & \mathbb{P}\left(\left\{ \frac{1}{\mathrm{SINR}_{\mathrm{canonical}}}<\frac{1}{\gamma}\right\} \right)\\
 & \overset{\left(a\right)}{=} & \int_{x=0}^{\frac{1}{\gamma}}\int_{\omega=-\infty}^{\infty}\Phi_{\frac{1}{\mathrm{SINR}_{\mathrm{canonical}}}}\left(\omega\right)\mathrm{e}^{-i\omega x}\frac{d\omega}{2\pi}dx,
\end{eqnarray*}
where $\left(a\right)$ is obtained by rewriting the c.d.f. of $\frac{1}{\mathrm{SINR}_{\mathrm{canonical}}}$
in terms of the characteristic function of $\frac{1}{\mathrm{SINR}_{\mathrm{canonical}}},$
where the inner integration computes the p.d.f. of $\frac{1}{\mathrm{SINR}_{\mathrm{canonical}}},$
and the outer integration gives the c.d.f. at $\frac{1}{\gamma}.$
When $\gamma=0,$ the above event occurs with probability 1, and otherwise,
it is expressed in terms of the integration in $\left(\ref{eq:ctoINtailProb}\right)$
which is obtained by exchanging the order of integrations in $\left(a\right),$
which is valid in this case, and then evaluating the integral w.r.t.
$x.$ In the rest of this section, we derive the expression for $\Phi_{\frac{1}{\mathrm{SINR}_{\mathrm{canonical}}}}\left(\omega\right),$
by first noting that $\mathrm{SINR}_{\mathrm{canonical}}=\frac{R_{1}^{-1}}{\sum_{k=2}^{\infty}R_{k}^{-1}+\eta}.$
\begin{eqnarray*}
 &  & \Phi_{\frac{1}{\mathrm{SINR}_{\mathrm{canonical}}}}\left(\omega\right)\overset{\left(a\right)}{=}\mathbb{E}_{R_{1}}\left[\Phi_{\left.\frac{1}{\mathrm{SINR}_{\mathrm{canonical}}}\right|R_{1}}\left(\left.\omega\right|R_{1}\right)\right]\overset{\left(b\right)}{=}\mathbb{E}_{R_{1}}\left[\mathrm{e}^{i\omega\eta R_{1}}\Phi_{\left.\frac{\sum_{k=2}^{\infty}R_{k}^{-1}}{R_{1}^{-1}}\right|R_{1}}\left(\left.\omega\right|R_{1}\right)\right]\\
 &  & =\mathbb{E}_{R_{1}}\left[\mathrm{e}^{i\omega\eta R_{1}}\Phi_{\left.\sum_{k=2}^{\infty}R_{k}^{-1}\right|R_{1}}\left(\left.\omega R_{1}\right|R_{1}\right)\right]\overset{\left(c\right)}{=}\mathbb{E}_{R_{1}}\left[\mathrm{e}^{i\omega\eta R_{1}}\mathbb{E}\left[\prod_{k=2}^{\infty}\left.\mathrm{e}^{i\omega R_{1}R_{k}^{-1}}\right|R_{1}\right]\right]\\
 &  & \overset{\left(d\right)}{=}\mathbb{E}_{R_{1}}\left[\mathrm{e}^{i\omega\eta R_{1}}\cdot\exp\left(-\int_{r=R_{1}}^{\infty}\left(1-\mathrm{e}^{i\omega R_{1}r^{-1}}\right)\lambda\left(r\right)dr\right)\right],
\end{eqnarray*}
where $\left(a\right)$ is obtained due to the properties of expectation,
and $R_{1}$ is the random variable for the distance of the closest
BS from the origin, $\left(b\right)$ is obtained by using the properties
of the characteristic functions and noting that in $\frac{1}{\mathrm{SINR}_{\mathrm{canonical}}}=\frac{\sum_{k=2}^{\infty}R_{k}^{-1}+\eta}{R_{1}^{-1}},$
conditioned on $R_{1}$, the term $\frac{\eta}{R_{1}^{-1}}$ is a
constant and hence separates out as $\mathrm{e}^{i\omega\eta R_{1}}$
from the original conditional characteristic function expression in
$\left(a\right),$ $\left(c\right)$ is obtained by rewriting the
exponential of summation in the characteristic function term in $\left(b\right)$
as a product of exponentials, $\left(d\right)$ is obtained by first
noting that conditioned on $R_{1},$ the events in the two disjoint
regions $\left[0,\ R_{1}\right]$ and $\left(R_{1},\ \infty\right)$
are independent of each other, and hence by the Restriction theorem
\cite[Page 17]{Kingman1993}, all the points beyond $R_{1},$ represented
by the set $\left\{ R_{k}\right\} _{k=1}^{\infty}$ can be regarded
to be associated with a Poisson point process in 1-D restricted to
the region $\left(R_{1},\infty\right),$ and with a density function
$\lambda\left(r\right).$ As a result, now we can apply Campbell's
theorem \cite[Page 28]{Kingman1993} to the inner expectation in $\left(c\right)$
to obtain $\left(d\right),$ which is further simplified to obtain
$\left(\ref{eq:charfunitoc}\right).$

\subsection{\label{sub:proofsinrTailProbGt1}Proof for Theorem \ref{thm:sinrTailProbGt1} }

Here, we derive the expression for the tail probability of SINR for
values greater than or equal to 1. Due to \cite[Lemma 1]{Dhillon2012},
there exists a unique BS within the 1-D SCS such that $\gamma\ge1$
holds true. Let the index of this unique BS be $i.$ The expression
for the tail probability of SINR is derived as follows.
\begin{eqnarray*}
\mathbb{P}\left(\left\{ \mathrm{SINR}>\gamma\right\} \right) & \overset{\left(a\right)}{=} & \mathbb{P}\left(\left\{ \frac{\Psi_{i}R_{i}^{-1}}{\sum_{j=1,\ j\ne i}^{\infty}\Psi_{j}R_{j}^{-1}+\eta}>\gamma\right\} \right)\\
 & \overset{\left(b\right)}{=} & \mathbb{E}\left[\exp\left(-\eta\gamma R_{i}\right)\prod_{j=1,\ j\ne i}^{\infty}\exp\left(-\gamma R_{i}\Psi_{j}R_{j}^{-1}\right)\right]\\
 & \overset{\left(c\right)}{=} & \mathbb{E}\left[\exp\left(-\eta\gamma R_{i}\right)\exp\left(-\int_{r=0}^{\infty}\left(1-\mathbb{E}_{\Psi}\left[\mathrm{e}^{-\gamma R_{i}\Psi r^{-1}}\right]\right)\bar{\lambda}\left(r\right)dr\right)\right]\\
 & \overset{\left(d\right)}{=} & \mathbb{E}\left[\exp\left(-\eta\gamma R_{i}\right)\exp\left(-\int_{r=0}^{\infty}\left(1-\frac{1}{1+\gamma R_{i}r^{-1}}\right)\bar{\lambda}\left(r\right)dr\right)\right],
\end{eqnarray*}
where $\left(a\right)$ is the expression for the tail probability
of SINR of the 1-D SCS with BS density $\bar{\lambda}\left(r\right)$
for which $\left\{ R_{j}\right\} _{j=1}^{\infty}$ is the set of distances
of BSs from the MS and `$i$' is the index of the unique BS that can
satisfy the constraint $\left\{ \mathrm{SINR}>\gamma\right\} ,$ $\left(b\right)$
is obtained by evaluating the expectation w.r.t. $\Psi_{i}$ and the
expectation operator $\mathbb{E}$ is w.r.t. to all other random variables
in $\left(a\right)$, $\left(c\right)$ is obtained by first conditioning
w.r.t. $R_{i}$ and noting that the Palm distribution of the BSs represented
by $\left\{ R_{j}\right\} _{j=1,\ j\ne i}^{\infty}$ given a BS at
$R_{i}$ is still a Poisson point process with density function $\bar{\lambda}\left(r\right),$
then applying the Marking theorem \cite[Page 55]{Kingman1993} and
Campbell's theorem \cite[Page 28]{Kingman1993} where $\Psi$ is the
unity mean exponential random variable, $\left(d\right)$ is obtained
by evaluating the expectation in $\left(c\right),$ and finally $\left(\ref{eq:sinrTailProbGt1}\right)$
is obtained by simplifying $\left(d\right).$

\subsection{\label{sub:proofFewBSApproxTheorem}Proof for the Few-BS Approximation
Theorem (Theorem \ref{thm:fewBSApproxTheorem})}

First, using Corollary \ref{cor:fewBSMeanCor}, $\mathrm{SIR}_{2}=\frac{KR_{1}^{-\varepsilon}}{\tilde{P_{I}}\left(2\right)}$,
with $\tilde{P_{I}}\left(2\right)=KR_{2}^{-\varepsilon}\left(1+\frac{\lambda_{0}b_{l}}{\varepsilon-l}R_{2}^{l}\right)$.
Next, notice that the event $\left\{ \mathrm{SIR}_{2}>\gamma\right\} $
is equivalent to the joint event $\left\{ R_{1}\le R_{2},\ R_{1}<\left(\frac{\gamma\tilde{P}_{I}\left(2\right)}{K}\right)^{-\frac{1}{\varepsilon}}\right\} $
and thus, $\mathbb{P}\left(\left\{ \frac{C}{I_{2}}>\gamma\right\} \right)=\mathbb{P}\left(\left\{ R_{1}\le\mathrm{min}\left(R_{2},\left(\frac{\gamma\tilde{P}_{I}\left(2\right)}{K}\right)^{-\frac{1}{\varepsilon}}\right)\right\} \right)$,
where
\begin{eqnarray*}
\mathrm{min}\left(R_{2},\left(\frac{\gamma\tilde{P}_{I}\left(2\right)}{K}\right)^{-\frac{1}{\varepsilon}}\right) & = & \begin{cases}
\left(\frac{\gamma\tilde{P}_{I}\left(2\right)}{K}\right)^{-\frac{1}{\varepsilon}}, & \gamma\ge1\\
\left(\frac{\gamma\tilde{P}_{I}\left(2\right)}{K}\right)^{-\frac{1}{\varepsilon}}, & \gamma<1,\ R_{2}>\left(\frac{l\times u\left(\gamma\right)}{\lambda_{0}b_{l}}\right)^{\frac{1}{l}}\\
R_{2} & ,\gamma<1,\ R_{2}\le\left(\frac{l\times u\left(\gamma\right)}{\lambda_{0}b_{l}}\right)^{\frac{1}{l}}
\end{cases}.
\end{eqnarray*}
 Finally, $\left(\ref{eq:few-bs-approx}\right)$ is obtained using
the joint p.d.f., $f_{R_{1},R_{2}}\left(r_{1},r_{2}\right)=\left(\lambda_{0}b_{l}\right)^{2}\left(r_{1}r_{2}\right)^{l-1}\exp\left(-\frac{\lambda_{0}b_{l}}{l}r_{2}^{l}\right),$
$0\le r_{1}\le r_{2}\le\infty,$ due to the properties of Poisson
point processes.

\subsection{\label{sub:proofSIRConvergesTo0}Proof for Theorem \ref{thm:SIRConvergesTo0}}

Let us consider the probability of the event that the interference
due to all the BSs beyond the signal BS at a given distance $R_{1}$
is below a certain value, say, $\delta,$ for the case $\varepsilon=2.$
\begin{eqnarray*}
 &  & \mathbb{P}\left(\left\{ \left.\sum_{k=2}^{\infty}R_{k}^{-2}\le\delta\right|R_{1}\right\} \right)=\mathbb{P}\left(\left\{ \left.\mathrm{e}^{-s\sum_{k=2}^{\infty}R_{k}^{-2}}\ge\mathrm{e}^{-s\delta}\right|R_{1}\right\} \right)\\
 &  & \overset{\left(a\right)}{\le}\mathrm{e}^{s\delta}\mathbb{E}\left[\left.\mathrm{e}^{-s\sum_{k=2}^{\infty}R_{k}^{-2}}\right|R_{1}\right]\overset{\left(b\right)}{=}\mathrm{e}^{s\delta}\mathrm{e}^{-\lambda\int_{r=R_{1}}^{\infty}\left(1-\mathrm{e}^{-sr^{-2}}\right)2\pi rdr}\\
 &  & \overset{\left(c\right)}{=}\mathrm{e}^{s\delta}\mathrm{e}^{\lambda\int_{r=R_{1}}^{\infty}\sum_{k=1}^{\infty}\frac{\left(-sr^{-2}\right)^{k}}{k!}2\pi rdr}\\
 &  & =\mathrm{e}^{s\delta}\mathrm{e}^{-\lambda s2\pi\cdot\left.\log\left(r\right)\right|_{r=R_{1}}^{\infty}+\lambda2\pi\sum_{k=2}^{\infty}\frac{\left(-s\right)^{k}}{k!}\frac{\left(R_{1}^{2-k\varepsilon}\right)}{k\varepsilon-2}}\\
 &  & =\mathrm{e}^{s\delta}\times0\times\mathrm{e}^{\alpha\left(R_{1}\right)}=0,
\end{eqnarray*}
where $\left(a\right)$ is obtained by applying Markov's inequality,
$\left(b\right)$ is obtained by applying Campbell's theorem to the
homogeneous Poisson point process defined in the 2-D plane beyond
$R_{1}$ from the origin, $\left(c\right)$ is obtained after the
Taylor's series expansion of the exponential function in $\left(b\right),$
and finally the result is obtained by noting that the exponential
of a sum of functions is a product of exponential and by showing that
one of the terms in the product is 0 while the others evaluate to
a finite number.

As a result,
\begin{eqnarray*}
\mathbb{P}\left(\left\{ SIR>\gamma\right\} \right) & = & \mathbb{E}_{R_{1}}\left[\mathbb{P}\left(\left\{ \left.\sum_{k=2}^{\infty}R_{k}^{-\varepsilon}<\left(\gamma R_{1}^{\varepsilon}\right)^{-1}\right|R_{1}\right\} \right)\right]\\
 & = & 0,\ \forall\ \gamma\ge0.
\end{eqnarray*}
 and hence we have proved the result.

\subsection{Simulation Methods\label{sec:algo}}

In this section, the details of simulating the SCS are presented.
A single trial in simulating the BS placement for the 1-D SCS with
BS density function $\lambda(r)$ in the region of interest which
is a subset of the 1-D plane denoted by $S$, involves the following
steps:

1) Generate a random number $M$, according to a Poisson distribution
with mean $\int_{S}\lambda\left(s\right)ds,$ which is the number
of BSs to be placed in $S$ for the given trial.

2)\textit{ BS placement:} For \textit{homogeneous} 1-D SCS, generate
$M$ random numbers according to a uniform distribution in the range
of $S$. If $\lambda(s)$ does not correspond to a \textit{homogeneous}
1-D SCS, if $\lambda_{\mathrm{max}}=\underset{s\in S}{\sup}\lambda(s)$,
then general a random number $y$ which is uniformly distributed in
the range $\left[0,\lambda_{\max}\right]$ and another random number
$x$ according to a uniform distribution in the range of $S$. A BS
is placed uniformly at $x$, only if $y<\lambda_{0}(x)$. This process
is repeated until $M$ BS are placed.

3) Compute the received power at the MS for each BS using the path-loss
exponent $\varepsilon$. The fading in the SCS is incorporated by
multiplying each of the received powers with i.i.d. random number
generated according to the distribution of the fading factor. Finally,
SINR at the MS corresponding to this trial, is computed according
to (\ref{eq:SINRexpression}).

For all the simulations in this paper $T=100,000$ trials are used
unless specified otherwise.\bibliographystyle{IEEEtran}
\bibliography{ShotgunCellularSystemAfterReview4AcceptedChanges}

\end{document}